\documentclass[letterpaper, 10 pt, conference]{ieeeconf}  

\IEEEoverridecommandlockouts                              

\usepackage{amsmath}
\usepackage{amsfonts}
\usepackage{amssymb}
\usepackage{amsthm}
\usepackage{dsfont}
\let\labelindent\relax
\usepackage{enumitem}
\usepackage{centernot}
\usepackage{flexisym}
\usepackage{xcolor}
\usepackage{multicol}
\usepackage{acro}
\usepackage{booktabs}
\usepackage{tabularx}
\usepackage{graphicx}
\usepackage{subcaption}
\usepackage{eurosym}

\newtheorem{theorem}{Theorem}
\newtheorem{lemma}{Lemma}
\newtheorem{proposition}{Proposition}

\theoremstyle{definition}
\newtheorem{definition}{Definition}

\newtheorem{problem}{Problem}
\newtheorem{remark}{Remark}

\DeclareMathOperator{\diag}{diag}
\DeclareMathOperator{\col}{col}

\DeclareMathAlphabet{\vt}{U}{bbold}{m}{n}
\newcommand{\norm}[1]{\ensuremath{\left\| #1 \right\|}}

\DeclareAcronym{ce}{
	short = CE,
	long  = Competitive Equilibrium,
	sort  = C,
}
\DeclareAcronym{sce}{
	short = SCE,
	long  = Socially acceptable Competitive Equilibrium,
	sort  = S,
}

\title{\LARGE \bf
Towards a Socially Acceptable Competitive Equilibrium  \\ in Energy Markets
}

\author{Koorosh Shomalzadeh and Nima Monshizadeh
\thanks{This research was made possible with the support of the Dutch Topsector Energy-subsidy of the Ministry of Economic Affairs and Climate, under the grant number TSYS2221002.}
\thanks{The authors are with {Engineering and Technology Institute Groningen, University~of~Groningen, Nijenborgh~4, 9747 AG Groningen, The Netherlands}. Email: {\tt\small k.shomalzadeh@rug.nl, n.monshizadeh@rug.nl}.}%
}

\begin{document}

\maketitle
\thispagestyle{empty}
\pagestyle{empty}

\begin{abstract}

This paper addresses the problem of energy sharing between a population of price-taking agents who adopt decentralized primal-dual gradient dynamics to find the \ac{ce}. Although the \ac{ce} is efficient, it does not ensure fairness and can potentially lead to high prices. As the agents and market operator share a social responsibility to keep the price below a certain socially acceptable threshold, we propose an approach where the agents modify their utility functions in a decentralized way. We introduce a dynamic feedback controller for the primal-dual dynamics to steer the agents to a \ac{sce}. We demonstrate our theoretical findings in a case study.
\end{abstract}

\acresetall
\section{Introduction} \label{sec:intro}
The market aspect of smart grids has been studied extensively in the recent years. General equilibrium theory has been utilized to understand and devise suitable pricing mechanisms. The notion of equilibrium in an electricity market refers to a condition where a market price is established through competition. Analogous to microeconomics principles, a range of competition models have been introduced for electricity markets, with the \ac{ce} receiving particular attention \cite{nguyen2011walrasian}. The emphasis on the \ac{ce} stems from its widely accepted status as a benchmark for market efficiency \cite{mas1995microeconomic}. When the behavior of a market closely aligns with the predictions of the \ac{ce}, it is considered to be functioning well. 
Through the \ac{ce}, optimal pricing of electricity is achieved, maximizing payoffs for all agents while matching the supply and demand. 

The emergence of information technology and communication infrastructures in the paradigm of smart grids has enabled the development of decentralized and distributed algorithms for finding the \ac{ce} without sharing the sensitive and private information of market participants \cite{li2020transactive}. In particular, primal-dual gradient dynamics, due to their scalability, have been used in multiple works to determine the \ac{ce} \cite{nguyen2011walrasian,stegink2016unifying,papadaskalopoulos2013decentralized,knudsen2015dynamic}, and these dynamics will be our starting point. 

It is a well-established fact that the \ac{ce} can be deemed unfair according to virtually any conceivable definition of fairness \cite{fehr1999theory}. As the \ac{ce} depends on individual needs and resource availability, market dynamics may lead to disparities in outcomes and increased prices, typically favoring those with greater wealth and purchasing ability. Additionally, when the equilibrium price rises excessively, individuals may opt out of active market participation. This phenomenon has been particularly observed in electricity markets recently, leading to multiple instances of price spikes. For instance, on April 4, 2022, the French wholesale market experienced a price spike, reaching up to $4000$ euros per MWh for some hours \cite{gerlagh2022stabilizing}.

To minimize extreme price fluctuations in the \ac{ce}, the authors in \cite{wei2014competitive} propose adding a price penalty term to the system objective, enabling a trade-off between price volatility and economic efficiency. However, yet the market price may become very high, and supply and demand may not match.
To address these points, the notion of socially acceptable price is considered in \cite{chen2022competitive,salehi2023competitive}. A price is considered socially acceptable if it is deemed fair for everyone in the community. The agents participating in the market are requested to choose their utility functions from a predetermined set of socially admissible functions in order to ensure that the equilibrium price stays below a designated threshold, and thus socially acceptable. Nonetheless, this set is not unique, and if chosen without careful consideration, it can significantly impact the market's efficiency.

In this paper, we assume that agents engaged in an electricity market follow the primal-dual gradient dynamics to achieve the \ac{ce}. Drawing inspiration from concerns regarding the inequity of the \ac{ce} and the concept of socially acceptable pricing, we aim for relocating the equilibrium of the market to a socially acceptable one. In particular, by viewing the primal-dual algorithm as a to-be-controlled system, we design a control strategy that can suitably modify the utility functions such that the resulting market price does not exceed a certain threshold, without significantly compromising the market's efficiency. We note that this viewpoint results in a setup that is significantly different than those in \cite{chen2022competitive}, where the agents choose their utility functions from a predetermined set. We analytically establish that the solutions of the closed-loop system converge to a CE of the market with a socially acceptable price. The proposed control algorithm does not rely on knowing or sharing any potentially sensitive parameters of the agents.

The remainder of the paper is organized as follows. In Section~\ref{sec:psf}, we formulate our problem and review some equilibrium concepts from microeconomics, as well as the primal-dual gradient dynamics to reach these equilibria. 
The main results of the paper are provided in Section~\ref{sec:mr}. 
Section~\ref{sec:cs} presents a case study demonstrating the relevance and effectiveness of the proposed control algorithm. Concluding remarks and future research directions are stated in Section~\ref{sec:con}. The proofs for the technical lemmas and propositions can be found in  Appendix.

\noindent \textit{Notation:}
We denote the set of real numbers by $\mathbb{R}$ and the set of nonnegative real numbers by $\mathbb{R}_+$.
We include the dimension of a set as a superscript, whenever needed.
 We use $\vt{1}$ to denote a vector of all ones and $\vt{0}$ to denote a vector/matrix of all zeros. The identity matrix is denoted by $I$.  
 Given a set $\mathcal{N}:=\{1,2,...,N\}$, the column vector obtained by stacking scalars $x_1, x_2, \dots, x_N$ is denoted by $\col(x_n)_{n\in\mathcal{N}}$. Moreover,  $\diag(x_n)_{n\in\mathcal{N}}$ denotes the diagonal matrix with scalars $x_1, x_2, \dots, x_N$ on its diagonal.
 For vectors $x$ and $y$, we write $0 \le x \perp y \ge 0$ if $x,y \ge 0$ and $x^\top y=0$.
For scalars $x$ and $y$, the operator $[x ]_{y}^{+}$
is defined as 
\begin{equation} \label{equ:operator}
    [x ]_{y}^{+}=
    \begin{cases}
     x, \quad &\mathrm{if} \ y>0, \\
     \max\{0,x\}, &\mathrm{if} \ y=0.
    \end{cases}
\end{equation}

\section{Problem Statement and Formulation} \label{sec:psf}
\subsection{Market model}
We consider a population $\mathcal{N}:=\{1,2,\dots, N\}$ of price-taking agents that participate in an electricity market. Each agent $i \in \mathcal{N}$ possesses a renewable energy source which generates $a_i \in \mathbb{R}_+$ electricity with zero marginal cost.
The agents are considered to be in an islanded microgrid, and thus no external energy supplier exists. Moreover, each agent $i\in \mathcal{N}$ is associated with a strictly concave utility function $f_i(x_i)$ that depends on their consumption level $x_i$. Given a uniform price $\lambda$, the agents are willing to sell their excess supply ($a_i-x_i \ge 0$) or buy their demand shortage ($a_i-x_i\le 0$) through a network. We can capture the goal of agent $i \in \mathcal{N}$ by the following optimization problem:
\begin{equation} \label{opt:agent_i}
    \min_{x_i } \quad -f_i(x_i)+\lambda(x_i-a_i) 
\end{equation}
We denote the demand profile vector of all agents by $x=\col(x_i)_{i\in \mathcal{N}}$. 
The equilibria of interest in this market are defined below:
\begin{definition}[\ac{ce}] \label{def:ce} A demand-price pair $(\bar x,\bar \lambda)$ is a \ac{ce} of the market  if 
\begin{enumerate}
    \item For each agent $i \in \mathcal{N}$, $\bar x_i$ is the optimal solution to problem \eqref{opt:agent_i} with $\lambda:=\bar \lambda$.
    \item The supply-demand matching condition holds, i.e.,
    $\sum_{i\in \mathcal{N}} \bar x_i= \sum_{i\in \mathcal{N}} a_i$.
\end{enumerate}
\end{definition}

\begin{definition}[Social welfare equilibrium]
A demand-price pair $(\bar x,\bar \lambda)$ is a social welfare equilibrium of the market if $(\bar x,\bar \lambda)$ satisfies  the KKT conditions of the social welfare optimization problem
\begin{equation} \label{opt:sw}
    \begin{aligned}
        &\min_{x} && -\sum_{i \in \mathcal{N}}f_i(x_i) \\
        & \textrm{s.t.} &&\sum_{i \in \mathcal{N}} x_i= \sum_{i \in \mathcal{N}} a_i,
    \end{aligned}
\end{equation}
i.e.,
\begin{equation}
\begin{aligned} \label{KKT:sw}
       -&\nabla_{x_i} f_i(\bar x_i)+\bar \lambda =0, \quad  \forall i \in \mathcal{N}, \\
      &\sum_{i \in \mathcal{N}} \bar x_i= \sum_{i \in \mathcal{N}} a_i.
\end{aligned}
\end{equation}
\end{definition}
For (strictly) concave utility functions, the aforementioned two notions of equilibrium coincide: 
\begin{lemma}[{\cite[Theorem~1]{li2015demand}}] \label{lem:CE&SW}
    Suppose that $f_i(x_i)$ is strictly concave for all $i \in \mathcal{N}$. Then, the \ac{ce} coincides with the social welfare equilibrium and is unique.
\end{lemma}
\subsection{Decentralized \ac{ce} seeking}
Given the result of Lemma~\ref{lem:CE&SW}, one can find the \ac{ce} by solving the social welfare optimization problem~\eqref{opt:sw}. Here, we restrict our attention to linear-quadratic utility functions of the form
\begin{equation} \label{eq:qu}
    f_i(x_i):=-\frac{1}{2}q_ix_i^2-c_ix_i,
\end{equation}
where $q_i>0$ and $c_i \le 0$.
To reach the \ac{ce}, we assume that the agents follow a decentralized primal-dual gradient dynamics \cite{arrow1958studies} arising from the social welfare optimization problem \eqref{opt:sw}.
Our interest in this algorithm stems from its capacity for parallel implementation, which scales effectively with the size of the energy markets \cite{li2020transactive}. Additionally, its widespread adoption highlights its practical applicability and effectiveness in addressing complex market dynamics \cite{nguyen2011walrasian,stegink2016unifying,papadaskalopoulos2013decentralized,knudsen2015dynamic}.
In particular, under the following decentralized primal-dual gradient dynamics, the agents' actions $x$ and the resulting price $\lambda$ converge to the \ac{ce} of the market:
\begin{subequations} \label{dyn:pd1}
    \begin{align}\label{dyn:pd1a}
        \dot x &= -Qx-c-\rho, \\ \label{dyn:pd1b}
        \dot \rho &= x-a-\epsilon,\\ \label{dyn:pd1c}
        \dot \epsilon &= \rho-\vt{1} \lambda,\\ \label{dyn:pd1d}
        \dot \lambda &=\vt{1}^{\top} \epsilon,
    \end{align}
\end{subequations}
where $Q:=\diag(q_i)_{i\in \mathcal{N}}$, $c:=\col(c_i)_{i \in \mathcal{N}}$ and $a:=\col(a_i)_{i \in \mathcal{N}}$. 
For each agent, the subdynamics in \eqref{dyn:pd1b} provides a local estimate of the price and \eqref{dyn:pd1a} is the gradient dynamics for each agent's action under the local price $\rho_i$. We note that \eqref{dyn:pd1a} and \eqref{dyn:pd1b} are performed by each agent, whereas \eqref{dyn:pd1c} and \eqref{dyn:pd1d} are executed by the market operator. In case the vector $a$ is available to the operator, the primal-dual algorithm can be reduced by removing the $\rho$ and $\epsilon$ dynamics, and redefining the price dynamics as $\dot \lambda= \vt{1}^{\top} x -\vt{1}^{\top} a$.

\subsection{A fair \ac{ce}}
Although the \ac{ce} is Pareto efficient, it does not guarantee fairness \cite{fehr1999theory}. Namely, market dynamics may lead to unequal outcomes and high prices, favoring individuals with greater wealth and purchasing power.
Additionally, if the equilibrium price becomes excessively high, individuals might refrain from participating actively in the market, instead of contributing to its self-sustainability. As members exit the system, the achievable payoff for the remaining agents decreases. 
This motivates the following definition:

\begin{definition}[\ac{sce}] \label{def:sace}
Let $\lambda^{\mathrm{max}}$ be a socially acceptable price given a priori.
A \ac{ce} $(\bar x, \bar \lambda )$ is called \textit{socially acceptable} if $\bar \lambda \le \lambda^{\mathrm{max}}$.
\end{definition}

There is a shared responsibility among the agents and the operator to maintain prices at socially acceptable levels to ensure fairness and uphold the market's self-sustainability.
As a \ac{ce} depends on the utility functions of the agents, bringing price below the threshold $\lambda^{\mathrm{max}}$ necessitates adjustments in the utility functions of the agents. 
In particular, to accommodate for fairness, we modify the functions in \eqref{eq:qu} to  \textit{controllable} utility functions given by
\begin{equation} \label{eq:quwu}
    f_i(x_i,u_i):=-\frac{1}{2}q_ix_i^2-c_i(u_i)x_i=-\frac{1}{2}q_ix_i^2-(c_{0i}+u_i)x_i,
\end{equation}
where the utility parameter $c_i\equiv c_i(u_i)$ is split into a nominal value 
$c_{0i}$ and a controllable quantity $u_i$ correspond to the needed modifications to bring the market price below $\lambda^\max$.

For a given $u:= \col(u_i)_{i \in \mathcal{N}} \in \mathbb{R}^N$, the pair $(\bar x, \bar \lambda)$ is a \ac{sce} if and only if it satisfies
\begin{subequations} \label{eq:sace}
    \begin{alignat}{1} \label{eq:sace_a}
        & \lambda \le \lambda^{\mathrm{max}}, \\ \label{eq:sace_b}
        &Q  x+c_0+\vt{1}  \lambda+u=0, \\\label{eq:sace_c}
        &\vt{1}^\top  x =  \vt{1}^\top a,
    \end{alignat}
\end{subequations}
for $(x, \lambda)=(\bar x, \bar\lambda)$.
Here,  \eqref{eq:sace_a} guarantees that the price is lower than $\lambda^{\mathrm{max}}$, while \eqref{eq:sace_b} and \eqref{eq:sace_c} are the KKT conditions as in \eqref{KKT:sw}, written for the utility functions in  \eqref{eq:quwu}. Observe that different choices of $u$ result in different socially acceptable equilibria. 
Among such choices, we look for the  \textit{minimum} adjustment $u$ that leads to a socially acceptable equilibrium. 
This goal is captured by the following optimization problem:
\begin{equation} \label{opt:swsa}
    \begin{aligned} 
        &\min_{ x,  \lambda, u} \quad  &&  \frac{1}{2} \norm{u}^2 \\  
        & \textrm{s.t.} && \eqref{eq:sace}. 
    \end{aligned}
\end{equation}
We denote the solution to the optimization above by  $ (x^*,  \lambda^*, u^*)$, and we refer to it as the \textit{optimal} \ac{sce}\footnote{We will show later in Lemma \ref{lem:swsa_unique} that \eqref{opt:swsa} admits a unique optimal solution.}, bearing in mind that the pair $(x^*,\lambda^*)$ is a socially acceptable equilibrium as it satisfies \eqref{eq:sace}.
It should be noted that in problem \eqref{opt:swsa}, if $u^* = 0$, then the \ac{ce} (see Def.~\ref{def:ce}) is a socially acceptable one (see Def.~\ref{def:sace} ). Otherwise, when $u^* \neq 0$, it indicates that the \ac{ce} has led to a price higher than $\lambda^{\mathrm{max}}$, and the agents' utlity functions have to be modified by $u^*$ to reach a \ac{sce}.

Recall that by following the primal-dual dynamics in \eqref{dyn:pd1} the agents reach 
a \ac{ce}, but not necessarily a socially desired one. By setting $c=c_0+u$, the 
dynamics in \eqref{dyn:pd1} modifies to 
\begin{subequations} \label{dyn:pd2}
    \begin{align}\label{dyn:pd2a}
        \dot x &= -Qx-c_0-\rho-u, \\ \label{dyn:pd2b}
        \dot \rho &= x-a-\epsilon,\\ \label{dyn:pd2c}
        \dot \epsilon &= \rho-\vt{1} \lambda,\\ \label{dyn:pd2d}
        \dot \lambda &=\vt{1}^{\top} \epsilon.
    \end{align}
\end{subequations}
Our aim in the rest of the paper is to design a control law for the primal-dual dynamics in \eqref{dyn:pd2}, with the control input $u$, to achieve the optimal \ac{sce} in \eqref{opt:swsa}. Moreover, we impose a decentralized structure on this control protocol to avoid sharing potentially sensitive information of the agents. This gives rise to the following problem formulation.
\begin{problem} \label{pr:main}
Consider the dynamical system \eqref{dyn:pd2}, with state variables $(x,\rho, \epsilon, \lambda)$ and control input $u$. 
Design a decentralized dynamic controller such that the closed-loop system is asymptotically stable and the solution $(x(t), \lambda(t), u(t))$ converges to the optimal \ac{sce} $(x^*, \lambda^*, u^*)$, i.e., the optimal solution to \eqref{opt:swsa}.
\end{problem}

We note that due to the problem formulation outlined above, the nominal behavior of the agents in \eqref{dyn:pd1} remains in tact since $u$ solely serves as an external input to $\eqref{dyn:pd1}$. This ensures that the solution to Problem \ref{pr:main} is acceptable to the agents. Such an advantage is not present if the operator chooses to directly impose a completely new set of primal-dual dynamics solving \eqref{opt:swsa}. 
The next section is dedicated to addressing Problem \ref{pr:main}.

\section{Main Results} \label{sec:mr}

First, we establish a few properties of the optimization problem in \eqref{opt:swsa}.

\begin{lemma} \label{lem:swsa_unique}
    Consider the optimization problem \eqref{opt:swsa}.
    Then, the following statements hold:
    \begin{enumerate}[leftmargin=*, label=(\alph*)]
        \item The KKT conditions of \eqref{opt:swsa} reduce to 
        \begin{equation} \label{kkt:swsa_reduced}
            0 \le \vt{1}^{\top}(\phi(\lambda) -a) \perp  \lambda^{\mathrm{max}} - \lambda \ge 0,
        \end{equation}
         with
        \begin{equation} \label{eq:g}
            \phi(\lambda):=- Q^{-1} (\vt{1} \lambda +c_0).
        \end{equation}
        \item The linear complementarity relation \eqref{kkt:swsa_reduced} has a unique solution $\lambda=\lambda^*$.
        \item The optimization problem \eqref{opt:swsa} admits a unique primal-dual optimizer.
    \end{enumerate}
\end{lemma}

Next, to tackle Problem~\ref{pr:main}, we present an optimization problem derived from the dual of \eqref{opt:sw} by incorporating some modifications. Subsequently, we demonstrate that the solution to this modified dual problem has a one-to-one relation with the solution to \eqref{opt:swsa}. The modified dual problem will then be used to devise the decentralized control law solving Problem \ref{pr:main}.

\subsection{Modified dual problem} \label{sec:mr_1}
Consider the following optimization problem: 
\begin{equation} \label{opt:swv}
    \begin{aligned}
        &\min_{y} && \frac{1}{2}y^{\top}Qy+c_0^{\top}y \\ 
        & \mathrm{s.t.} && \vt{1}^{\top}y =\vt{1}^{\top} a.
    \end{aligned}
\end{equation}
Observe that the above problem is the vector form of \eqref{opt:sw} with $y\equiv x$, $f_i$ as in \eqref{eq:quwu} and $u_i=0$. 
We call the pair $(\bar y, \bar \lambda)$ a \textit{primal-dual optimizer} of \eqref{opt:swv} if it satisfies the KKT conditions of  \eqref{opt:swv}. Note that this optimizer is unique since \eqref{opt:swv} is strictly convex. 
Next, we write the so-called \textit{dual problem} of \eqref{opt:swv} given by
\begin{equation} \label{opt:swv_dual}
    \min_\lambda \quad \frac{1}{2} (\vt{1}^{\top}Q^{-1}\vt{1}) \lambda^2 +(Q^{-1}c_0+a)^{\top}\vt{1} \lambda
\end{equation}
\begin{lemma} \label{lem:dual1}
    Consider the optimization problems \eqref{opt:swv} and \eqref{opt:swv_dual}. If $\bar \lambda$ is the optimal solution to \eqref{opt:swv_dual}, then  $(\bar y, \bar\lambda)$ is the primal-dual optimizer to \eqref{opt:swv} with $\bar y= -Q^{-1} (\vt{1} \bar \lambda+c_0)$.
    Conversely, if $(\bar y, \bar\lambda)$ is the primal-dual optimizer to \eqref{opt:swv}, then $\bar \lambda$ is the optimal solution to \eqref{opt:swv_dual}.
\end{lemma}

As the dual variable $\lambda$ is a decision variable in \eqref{opt:swv_dual}, to restrict the price, we can add the constraint $\lambda \le \lambda^{\mathrm{max}}$ to this problem, which yields
\begin{equation} \label{opt:swv_dual_mod}
    \min_{\lambda\le \lambda^{\max}} \quad \frac{1}{2} (\vt{1}^{\top}Q^{-1}\vt{1}) \lambda^2 +(Q^{-1}c_0+a)^{\top}\vt{1} \lambda
\end{equation}
We now attempt to dualize once more the modified problem \eqref{opt:swv_dual_mod}. This leads to the following  problem:
\begin{equation} \label{opt:swv_primal_mod}
    \begin{aligned}
        &\min_{y,s} && \frac{1}{2}y^{\top}Qy+c_0^{\top}y +\lambda^{\max} s\\ 
        & \mathrm{s.t.} && \vt{1}^{\top}y-s =\vt{1}^{\top} a,\\
        &&& s \ge 0,
    \end{aligned}
\end{equation}
with the addition of a new scalar slack variable $s$. We call $\big((\bar y, \bar s), (\bar \lambda, \bar \mu_s)\big)$ a \textit{primal-dual optimizer} of \eqref{opt:swv_primal_mod} if it satisfies the KKT conditions of  \eqref{opt:swv_primal_mod}, where $\bar \lambda$ and $\bar \mu_s$ denote the optimal values of the dual variables associated with the equality and nonnegativity constraint, respectively.
The duality  between \eqref{opt:swv_dual_mod} and \eqref{opt:swv_primal_mod} is established below:

\begin{lemma} \label{lem:dual2}
If $\bar \lambda$ is the optimal solution to \eqref{opt:swv_dual_mod}, then $\big((\bar y, \bar s), (\bar \lambda, \bar \mu_s)\big)$ is the primal-dual optimizer to \eqref{opt:swv_primal_mod} with $\bar y= \phi(\bar \lambda)$, $\bar s= \vt{1}^{\top}(\phi(\bar \lambda)-a)$, $\bar \mu_s=\lambda^{\max}-\bar \lambda$ and $\phi(\lambda)$ as in \eqref{eq:g}. Conversely, if $\big((\bar y, \bar s), (\bar \lambda, \bar \mu_s)\big)$ is the primal-dual optimizer to \eqref{opt:swv_primal_mod}, then 
$\bar \lambda$ is the optimal solution to \eqref{opt:swv_dual_mod}.
Moreover, the KKT conditions of both \eqref{opt:swv_dual_mod} and \eqref{opt:swv_primal_mod} reduce to \eqref{kkt:swsa_reduced}  with $\lambda=\bar \lambda$.
\end{lemma}

We conclude this subsection by establishing the relationship between the solutions to \eqref{opt:swsa} and \eqref{opt:swv_primal_mod}.

\begin{proposition} \label{pro:map}
 Let $\nu^*$ be the optimal dual variable corresponding to the constraint \eqref{eq:sace_c} and $\big((\bar y, \bar s), (\bar \lambda, \bar \mu_s)\big)$ be the primal-dual optimizer to \eqref{opt:swv_primal_mod}. Then, we have $\bar \lambda=\lambda^*$,
 \begin{equation} \label{eq:cov}
    \begin{bmatrix}
        \bar y \\ \bar s 
    \end{bmatrix}=
    M \begin{bmatrix}
         x^* \\ \nu^* 
    \end{bmatrix},
\end{equation}
with 
\begin{equation} \label{eq:M}
    M:=\begin{bmatrix}
        I & Q^{-2}\vt{1} \\ \vt{0} & \vt{1}^{\top} Q^{-2} \vt{1}
    \end{bmatrix}.
\end{equation}
Moreover, $M$ is nonsingular. 
\end{proposition}

In the next subsection, we utilize the aforementioned results to obtain the controller specified in Problem~\ref{pr:main}.

\subsection{Controller design}

We propose the following dynamics as the controller for the system \eqref{dyn:pd2}:
\begin{subequations} \label{dyn:cont}
    \begin{alignat}{1} \label{dyn:cont_a}
        \dot u& =  -Q^{-1} u -Q \pi-x - Q^{-1}(c_0+\lambda^{\mathrm{max}}\vt{1}),\\ \label{dyn:cont_b}
         \dot \pi & = Q u - \vt{1} \nu, \\ \label{dyn:cont_c}
         \dot \nu & = \vt{1}^{\top} \pi + \mu,\\  \label{dyn:cont_d}
        \dot  \mu & = [-\nu ]_{\mu}^{+}.
    \end{alignat}
\end{subequations}

We note that \eqref{dyn:cont} is a decentralized dynamic feedback controller with the input $x$ and output $u$. While  \eqref{dyn:cont_a} and \eqref{dyn:cont_b} are performed by each agent, equations \eqref{dyn:cont_c} and \eqref{dyn:cont_d} are executed by the market operator. The operator does not require any information on the private utility parameters $q_i$s and $c_{0i}$s. We postpone the discussion on how the controller dynamics can be derived based on the treatment in Subsection~\ref{sec:mr_1} towards the end of this subsection.

The next result shows that the closed-loop system, \eqref{dyn:pd2} and \eqref{dyn:cont}, possesses a unique equilibrium, and this equilibrium has the  desired properties specified in Problem \ref{pr:main}.

\begin{proposition} \label{pro:unique_equ}
The closed-loop system, comprising  \eqref{dyn:pd2} and \eqref{dyn:cont}, admits a unique equilibrium 
$(\bar x, \bar \rho, \bar \epsilon, \bar \lambda, \bar u, \bar \pi, \bar \nu, \bar \mu)$. Moreover, $\bar x=x^*$, $\bar\lambda=\lambda^*$, and $\bar u=u^*$, where $(x^*, \lambda^*, u^*)$ is the optimal \ac{sce} given by \eqref{opt:swsa}.
\end{proposition}

The next theorem establishes that the 
controller \eqref{dyn:cont} solves Problem~\ref{pr:main}.
\begin{theorem} \label{the:main}
The equilibrium of the closed-loop system, comprising  \eqref{dyn:pd2} and \eqref{dyn:cont}, is globally\footnote{Note that the domain of the variable $\mu$ is restricted to $\mathbb{R}_+$; see \eqref{equ:operator}.} asymptotically stable. In particular, $\lim_{t \to \infty} (x(t), \lambda(t), u(t))=( x^*,  \lambda^*, u^*)$, i.e.,  the optimal \ac{sce}.
\end{theorem} 

\begin{proof}
Let $\mathbf{x}:=\col ( x,  \rho, \epsilon,  \lambda,  u,  \pi,  \nu,  \mu)$ be a solution of the closed-loop system,   $ \mathbf{\bar x}:=\col (\bar  x, \bar  \rho,\bar  \epsilon,  \bar \lambda, \bar u,\bar  \pi,\bar  \nu, \bar \mu)$ be an equilibrium of it and $\mathbf{\tilde x}:=\mathbf{x}- \mathbf{\bar x}$.  Note that as a result of  Proposition~\ref{pro:unique_equ},  $ \mathbf{\bar x}$ is unique,  $\bar x=x^*$, $\bar\lambda=\lambda^*$ and $\bar u=u^*$.
We consider a quadratic Lyapunov function as 
$V(\mathbf{\tilde x}):=\frac{1}{2} \norm{\mathbf{\tilde x}}^2$.
The time-derivative of the evolution of $V$ along the solutions of the closed-loop system is computed as 
\begin{align*}
    &\dot V=\mathbf{\tilde x}^\top  \mathbf{\dot x}=
    \\ &{\tilde x}^{\top} \dot x +{\tilde \rho}^{\top} \dot \rho + {\tilde \epsilon}^{\top} \dot \epsilon +{\tilde \lambda} \dot \lambda + {\tilde u}^{\top} \dot u +
    {\tilde \pi}^{\top} \dot \pi + {\tilde \nu} \dot \nu + {\tilde \mu} [-\nu]_{\mu}^+.
\end{align*}
By adding and subtracting $\tilde \mu(-\nu)$, we can rewrite  $\dot V$ as
\begin{multline*}
    \dot V={\tilde x}^{\top} \dot x + {\tilde \rho}^{\top} \dot \rho + {\tilde \epsilon}^{\top} \dot \epsilon + {\tilde \lambda} \dot \lambda + {\tilde u}^{\top} \dot u +
    {\tilde \pi}^{\top} \dot \pi + {\tilde \nu} \dot \nu\\ + \tilde \mu(-\nu) +{\tilde \mu} \big( [-\nu]_{\mu}^+ - (-v) \big).
\end{multline*}
Next, we show that the last term on the right-hand side of the equality above is nonpositive. Suppose that $\mu > 0$. Then, $[-\nu]_{\mu}^+=-\nu$ and thus $\tilde \mu \big([-\nu]_{\mu}^+-(-\nu) \big)=0$. Now, suppose that $\mu=0$. Then, we have $[-\nu]_{\mu}^+\ge -\nu$ and $\tilde \mu=\mu-\mu^* \le 0$. This proves $\tilde \mu \big([-\nu]_{\mu}^+-(-\nu) \big) \le 0$.
Consequently, we have
\begin{equation} \label{pd3:lyp1}
    \dot V \le {\tilde x}^{\top} \dot x + {\tilde \lambda} \dot \lambda + {\tilde u}^{\top} \dot u +
    {\tilde \pi}^{\top} \dot \pi + {\tilde \nu} \dot \nu + \tilde \mu(-\nu).
\end{equation}
In addition, as $\mathbf{\bar x}$ is an equilibrium of the system, the following equalities hold:
\begin{equation} \label{pd3:lyp2}
    \begin{aligned}
         &-\tilde x^\top(-Q \bar x-c-\bar \rho -\bar u)=0, \\
         &-\tilde \rho^\top( \bar x-\bar \epsilon -a)=0,\\
         &-\tilde \epsilon^\top ( \bar \rho- \vt{1} \bar \lambda)=0,\\
         &-\tilde \lambda (\vt{1}^{\top} \bar \epsilon)=0, \\
        &-\tilde u^\top (-Q^{-1} \bar u -Q \bar \pi-\bar x - Q^{-1}(c+\lambda^{\mathrm{max}}\vt{1}))=0,\\
         &-\tilde \pi^\top (Q \bar u - \vt{1} \bar \nu)=0, \\
        &-\tilde \nu (\vt{1}^{\top} \bar \pi + \bar \mu)=0.
    \end{aligned}
\end{equation}

Moreover, it holds that 
\begin{equation} \label{pd3:lyp2p}
    -\tilde \mu(-\bar \nu) \ge 0.
\end{equation}
To see this, using \eqref{equ:cl0_h}, we note that 
for $\ \bar \mu>0$, we have $0=[-\bar \nu]_{\bar \mu}^{+}=-\bar \nu$ and consequently $-{\tilde \mu}(-\bar \nu) = 0$.  Otherwise, if $\ \bar \mu=0$, then $0=[-\bar \nu]_{\bar \mu}^{+} \ge -\bar \nu$ and $\tilde \mu =(\mu-\bar \mu) \ge 0$. 

Now, by adding the expressions in \eqref{pd3:lyp2} and \eqref{pd3:lyp2p} to the right hand side of \eqref{pd3:lyp1}, we get
$\dot V \le 
    \mathbf{\tilde x}^{\top} X \mathbf{\tilde x} =\mathbf{\tilde x}^{\top} X_{\mathrm{sym}} \mathbf{\tilde x}$
with $X_{\mathrm{sym}}:=\frac{X+X^\top}{2}$  and 
\begin{equation*}
\small
    X:=
    \begin{bmatrix}
        -Q & -I & \vt{0} & \vt{0} &-I &\vt{0} &\vt{0} &\vt{0} \\
        I &\vt{0} &-I &\vt{0} &\vt{0} &\vt{0} &\vt{0} &\vt{0} \\
        \vt{0} & I &\vt{0} &-\vt{1} &\vt{0} &\vt{0} &\vt{0} &\vt{0} \\
        \vt{0} &\vt{0} &\vt{1}^{\top} &\vt{0} &\vt{0} &\vt{0} &\vt{0} &\vt{0} \\
        -I &\vt{0} &\vt{0} &\vt{0} & -Q^{-1} & -Q &\vt{0} &\vt{0} \\
        \vt{0} &\vt{0} &\vt{0} &\vt{0} & Q & \vt{0} & -\vt{1} & \vt{0} \\
        \vt{0} &\vt{0} &\vt{0} &\vt{0} &\vt{0} & \vt{1}^\top & 0 & 1 \\
        \vt{0} &\vt{0} &\vt{0} &\vt{0} &\vt{0} & \vt{0} & -1 &0
    \end{bmatrix}
    .
\end{equation*}
It is easy to see that $X_{\mathrm{sym}}$ can be written as ${X_{\mathrm{sym}}=-B^\top B}$ with
$
B:=
\begin{bmatrix}
    Q^{\frac{1}{2}} & \vt{0} &\vt{0} &\vt{0}  &Q^{-\frac{1}{2}}& \vt{0}& \vt{0}& \vt{0} 
\end{bmatrix}
$.
Hence, $X_{\mathrm{sym}}$ is negative semidefinite and  $\dot V \le 0$.
The set of points satisfying $\dot V=0$ is given by ${\mathcal{S}:=\{\mathbf{\tilde x} |\, B \mathbf{\tilde x}= 0\}}$.
Then, by applying a Caratheodory variant of LaSalle’s invariance principle (see \cite{CHERUKURI201610}), we conclude that $\mathcal{S}$ does not contain any trajectory of the system, except the point $\mathbf{\tilde x}=0$. Hence, $\mathbf{x}$ converges to the unique equilibrium of the closed-loop system $\mathbf{\bar x}$ with $\bar x=x^*$, $\bar\lambda=\lambda^*$, and $\bar u=u^*$.
\end{proof}

\noindent \textit{Insights on the derivation of the controller dynamics \eqref{dyn:cont}:} 
First, we change the variables of \eqref{opt:swv_primal_mod} from $(y, s)$ to $(x, \nu)$ using the map 
    $\left[ \begin{smallmatrix}  y \\  s  \end{smallmatrix} \right]=
    M \left[ \begin{smallmatrix} x \\  \nu  \end{smallmatrix} \right]$ ,
    with $M$ as in \eqref{eq:M}. This change of variables leads to the optimization problem 
    \begin{subequations} \label{opt:after_cov}
    \begin{alignat}{2} \nonumber
        &\min_{x,\nu} \quad && \frac{1}{2}x^{\top}Qx+(c_0 +Q^{-1}\vt{1}  \nu)^\top x+ \\ \nonumber
        &&&\frac{1}{2} (Q^{-1}\vt{1}  \nu)^\top Q^{-1}(Q^{-1}\vt{1}  \nu) +\\  \label{opt:after_cov_a}
        &&&   (c_0+\lambda^{\mathrm{max}}\vt{1})^{\top}Q^{-1}(Q^{-1}\vt{1}  \nu) \\ \label{opt:after_cov_b}
        & \textrm{s.t.} && \vt{1}^{\top}x=\vt{1}^{\top}a, \\ \label{opt:after_cov_c}
        &&& \nu \ge 0.
    \end{alignat}
\end{subequations}
Note that as $M$ is nonsingular, 
$\left[ \begin{smallmatrix} x^* \\  \nu^*  \end{smallmatrix} \right] 
    =
    M^{-1}
    \left[ \begin{smallmatrix} \bar y \\ \bar s  \end{smallmatrix} \right]
    $ 
is the optimal solution to \eqref{opt:after_cov} (see \cite[Section~4.1.3]{boyd2004convex}). Moreover, since the problem \eqref{opt:swv_primal_mod} is convex, this change of variables does not change the optimal value of the dual variables. Thus, based on Proposition~\ref{pro:map}, the optimal value of the dual variable corresponding to the constraint \eqref{opt:after_cov_b} is $\bar \lambda= \lambda^*$.
Next, motivated by the fact that $u^*=Q^{-1} \vt{1} \nu^*$ (see \eqref{kkt:mpec_quad1_middle}), we substitute $Q^{-1}\vt{1} \nu$ in the objective function  \eqref{opt:after_cov_a} with $u$ and add  $Qu= \vt{1} \nu$ as a constraint to the optimization problem. Furthermore, to facilitate decentralized implementation, we replace the constraint ${\vt{1}^{\top}x=\vt{1}^{\top}a}$ with $x-a-\epsilon=0$ and $\vt{1}^{\top}\epsilon=0$. This gives us  the following problem:
\begin{subequations} \label{opt:after_cov&sub}
    \begin{alignat}{2} \nonumber
        &\min_{x,u, \nu, \epsilon } \quad &&  \frac{1}{2}x^{\top}Qx+(c_0+u)^{\top}x+\\
        &&&\frac{1}{2}u^{\top}Q^{-1} u+(c_0+\lambda^{\mathrm{max}}\vt{1})^{\top}Q^{-1}u \\ \label{opt:after_cov&sub_b} 
        & \textrm{s.t.} && \vt{1}^{\top}\epsilon=0,  \\ \label{opt:after_cov&sub_c}
        &&&  x-a-\epsilon=0,\\ \label{opt:after_cov&sub_d}
        &&&Q u=  \vt{1} \nu,\\ \label{opt:after_cov&sub_e}
        &&&\nu \ge 0.
    \end{alignat}
\end{subequations}
Now, let  $\lambda$, $\rho$, $\pi$ and $\mu$ be the dual variables corresponding to the constraint \eqref{opt:after_cov&sub_b}-\eqref{opt:after_cov&sub_e}, respectively. Then, the primal-dual gradient dynamics of this optimization problem coincide with the closed-loop system \eqref{dyn:pd2} and \eqref{dyn:cont}. In particular, the equilibrium of the closed-loop system as specified in Proposition~\ref{pro:unique_equ} coincides with solutions of the KKT conditions of problem \eqref{opt:after_cov&sub}. 
\begin{remark}
Theorem~\ref{the:main} establishes asymptotic stability of primal-dual dynamics associated with the optimization problem \eqref{opt:after_cov&sub}. While numerous studies have addressed the asymptotic stability of the equilibrium point in primal-dual gradient dynamics, these investigations predominantly rely on the assumption of strict or strong convexity of the objective function. Such investigations do not readily apply to the \textit{nonstrict} convex optimization problem \eqref{opt:after_cov&sub}. 
Nonetheless, despite lack of strict convexity, \eqref{opt:after_cov&sub} has a unique primal-dual optimizer which enables us to show asymptotic convergence of the solutions of the closed-loop system \eqref{dyn:pd2} and \eqref{dyn:cont} to its equilibrium point. 
\end{remark}

\section{Case study} \label{sec:cs}
We numerically demonstrate the merit of the \ac{sce} considered here, along with the controller's performance in achieving this equilibrium.  We consider a population $\mathcal{N}=\{1,2,3,4\}$ of agents participating in the market.
The data concerning the agents' parameters and generations can be found in Table~\ref{table:data}. It should be noted that the parameter $q_i$ signifies an agent's urgency or desire for energy, while $c_i$ indicates the consumption level at which the agent derives the most utility, represented by $\frac{-c_i}{q_i}$.
In the particular case we study here, agents $1$ and $2$ have relatively high demand but with low urgency, whereas agents $3$ and $4$ show low demand but with high urgency, as reflected in the chosen parameters in the first and third columns of the table. 
\setlength{\tabcolsep}{8pt} 
\renewcommand{\arraystretch}{1} 
\begin{table}[h]
\caption{Agents' data}
\centering
\begin{tabular}{@{}ccccc@{}}
\toprule
{agent} & \textbf{$q_i(\text{€}/\mathrm{kWh}^2)$}    & \textbf{$c_i (\text{€}/\mathrm{kWh})$} & \textbf{$-\frac{c_i}{q_i}(\mathrm{kWh})$} & \textbf{$a_i (\mathrm{kWh})$}  \\ \midrule
\textbf{$1$} &  1    &        -50        &        50        &     48                                                             \\
\textbf{$2$} & 1.5    &      -60          &      40          &    30                                                            \\
\textbf{$3$} & 10    &        -40        &         4       &     1.5                                                            \\
\textbf{$4$} & 20    &       -20         &          1      &      0.5                                                           \\ \bottomrule
\end{tabular}
\label{table:data}
\end{table}

We first determine the \ac{ce} of the market by solving the optimization problem \eqref{opt:sw}. This returns the values of the CE as
\[
 \big( (\bar x_1, \bar x_2, \bar x_3, \bar x_4 ), \bar \lambda\big)= \big( (41.74 , 34.5, 3.17, 0.59), 8.26 \big),
\]
which means that the market price at the CE escalates to $8.26 \, \text{€}/\mathrm{kWh}$. Additionally, despite the pressing need for energy, the \ac{ce} demand for agents 3 and 4 falls below their respective demands for maximum satisfaction. This discrepancy is particularly striking for agent 4, who, despite having the highest urgency, is allocated a demand of merely $0.59 \, \mathrm{kWh}$. This amount represents a notably small fraction of what would be required for agent 4's maximum satisfaction.



Next, we analyze the outcomes of the \ac{sce}. We set the socially acceptable  price, denoted by $\lambda^\max$, to $4 \, \text{€}/\mathrm{kWh}$. Figure~\ref{fig:sce} displays the \ac{sce} results, derived from the dynamics \eqref{dyn:pd2} and \eqref{dyn:cont}. The equilibrium is asymptotic stability as expected, and the SCE of the market is obtained as 
\begin{gather*}
    \big((x_1^*,x_2^*,x_3^*,x_4^* ),\lambda^*, (u_1^*,u_2^*,u_3^*,u_4^* ) \big)=\\
    \big( (40.69,34.98,3.55,0.79 ),4, (5.31,3.54,0.53,0.26) \big).
\end{gather*}
Furthermore, in contrast to the \ac{ce}, the \ac{sce} effectively allows agents 3 and 4 to secure a higher demand due to their acute urgency. Additionally, agent 1, characterized by the lowest urgency, consumes less compared to the \ac{ce}, showing the \ac{sce} capability to more equitably redistribute demand among agents.  

\begin{figure}[h]
  \centering
  \begin{subfigure}[b]{.75\linewidth}
    \includegraphics[width=\linewidth]{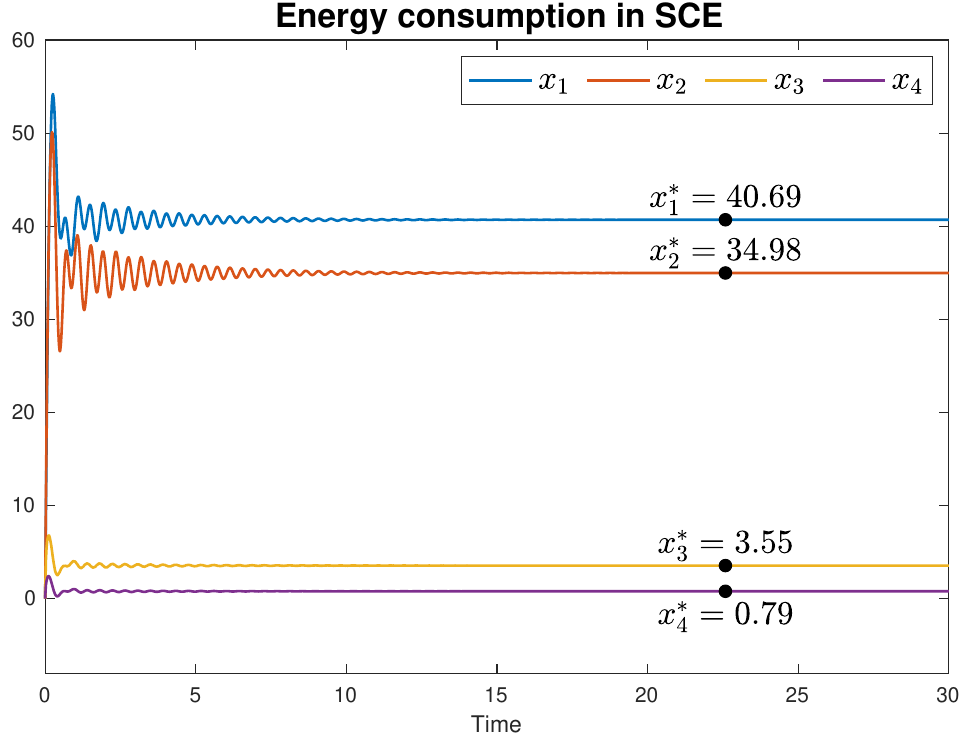}
      \end{subfigure}

  \begin{subfigure}[b]{.75\linewidth}
    \includegraphics[width=\linewidth]{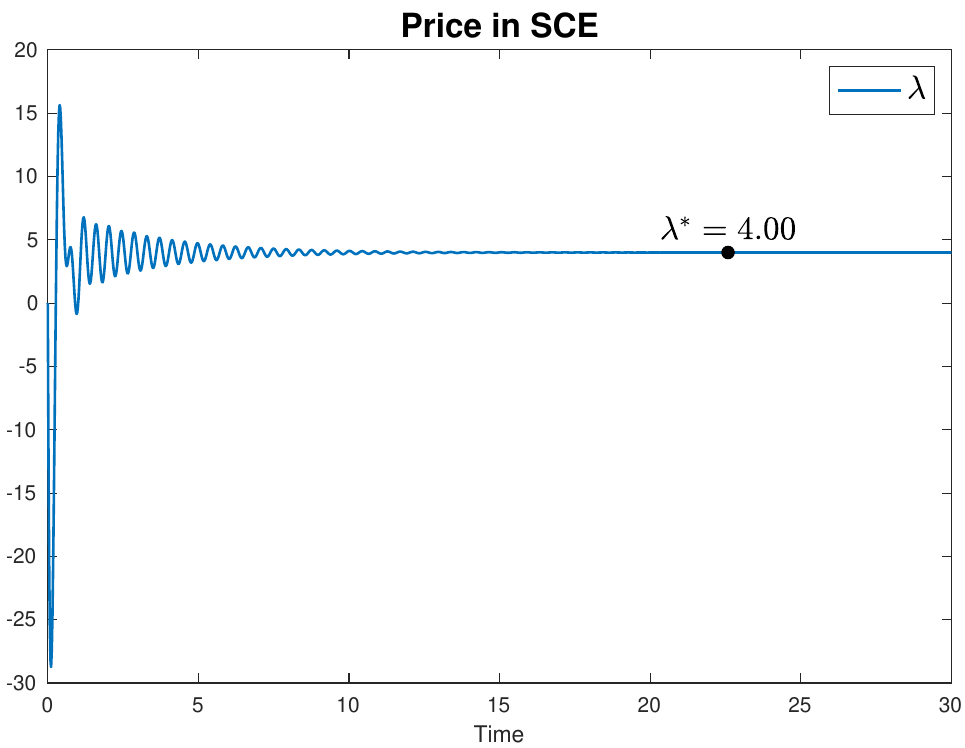}
  \end{subfigure}
  
  \begin{subfigure}[b]{.75\linewidth}
    \includegraphics[width=\linewidth]{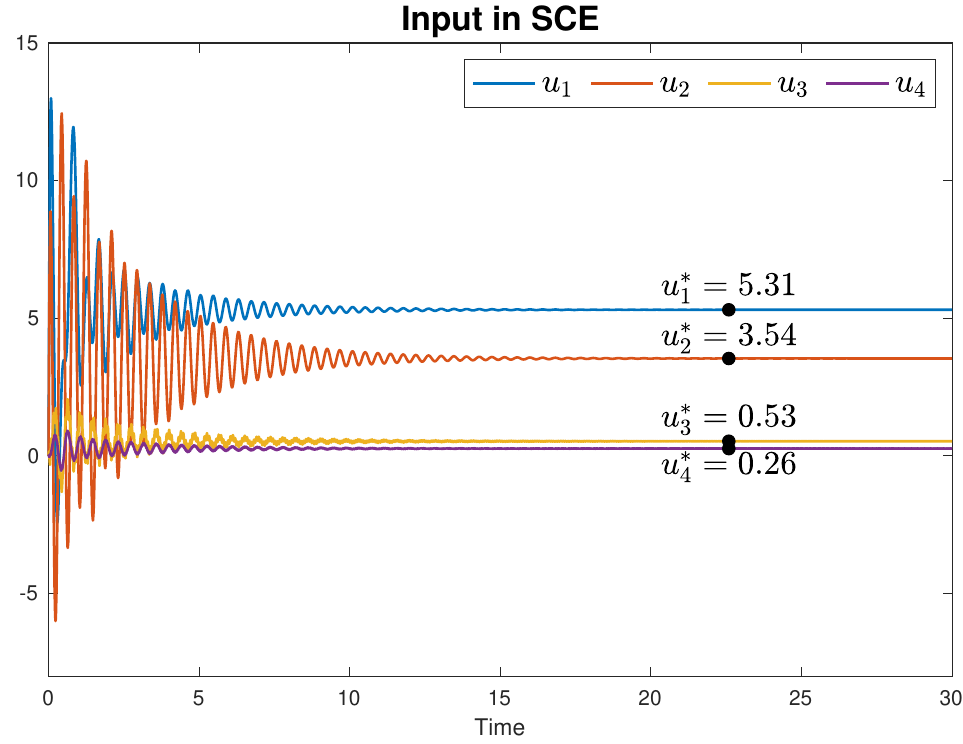}
  \end{subfigure}
  \caption{Trajectory of the closed-loop system \eqref{dyn:pd2} and \eqref{dyn:cont}, demonstrating convergence to the \ac{sce} with price $\lambda^{\max}=4$.}
  \label{fig:sce}
\end{figure}

\section{Conclusions} \label{sec:con}
In this work, we have studied the problem of energy sharing for a population of agents. While the \ac{ce} is known to be Pareto optimal, it might lead to high prices. By considering controllable utility functions for the agents, we have introduced a decentralized controller which enables the agents to converge asymptotically to the optimal \ac{sce} with a certain socially acceptable price. 

For future work, one avenue is to consider the inclusion of different market players, such as suppliers of non-renewable electricity. Additionally, as the physical grid constraints can change the consumption allocation of the agents, another direction involves taking these constraints into account.

\section*{APPENDIX: Proof of Lemmas and Propositions} \label{sec:app}
\begin{proof}[Proof of Lemma~\ref{lem:swsa_unique}]
Let $\pi_1^*$, $\pi_2^*$ and $\nu^*$  be the optimal  dual variables corresponding to the constraints \eqref{eq:sace_a}, \eqref{eq:sace_b} and \eqref{eq:sace_c}, respectively.
A primal-dual optimizer of \eqref{opt:swsa} satisfies the necessary and sufficient KKT conditions given by
\begin{subequations} \label{kkt:mpec_quad1_full}
    \begin{alignat}{1} \label{kkt:mpec_quad1_full_a}
       & u^*+ \pi_2^*=0, \\ \label{kkt:mpec_quad1_full_b}
        &Q \pi_2^*+\vt{1} \nu^*=0,\\ \label{kkt:mpec_quad1_full_c}
        & \pi_1^*+\vt{1}^{\top} \pi_2^*=0,\\ \label{kkt:mpec_quad1_full_d}
        &Q x^*+c_0+ u^*+\vt{1} \lambda^* =0,\\ \label{kkt:mpec_quad1_full_e}
        &\vt{1}^{\top} x^*-\vt{1}^{\top}a=0,\\ \label{kkt:mpec_quad1_full_f}
        &0 \le \pi_1^* \perp  \lambda^{\mathrm{max}} - \lambda^* \ge 0.
    \end{alignat}
\end{subequations}
Given \eqref{kkt:mpec_quad1_full_a}-\eqref{kkt:mpec_quad1_full_c}, we get 
\begin{equation} \label{kkt:mpec_quad1_middle}
    u^*=Q^{-1}\vt{1} \nu^*, \; \pi_1^*=\vt{1}^{\top} Q^{-1} \vt{1} \nu^*, \; \pi_2^*=-Q^{-1} \vt{1} \nu^*.
\end{equation}
Based on \eqref{kkt:mpec_quad1_middle}, we rewrite \eqref{kkt:mpec_quad1_full_d}-\eqref{kkt:mpec_quad1_full_f} as
\begin{subequations} \label{kkt:mpec_quad1}
    \begin{alignat}{1} \label{kkt:mpec_quad1_a}
        &Q x^*+c_0+Q^{-1}\vt{1}  \nu^*+\vt{1} \lambda^* =0,\\  \label{kkt:mpec_quad1_b}
        &\vt{1}^{\top} x^*-\vt{1}^{\top}a=0, \\  \label{kkt:mpec_quad1_c}
        &0 \le  \nu^* \perp \lambda^{\mathrm{max}}- \lambda^*  \ge 0.
    \end{alignat}
\end{subequations}
Moreover, based on \eqref{kkt:mpec_quad1_a} and \eqref{kkt:mpec_quad1_b}, we have
\begin{subequations} \label{eq:lambda_nu&x}
    \begin{alignat}{1}\label{eq:lambda_nu&x_a}
        &\nu^*=(\vt{1}^{\top}Q^{-2} \vt{1})^{-1} \vt{1}^{\top} (\phi(\lambda^*)-a), \\
        \label{eq:lambda_nu&x_b}
        &x^*=\phi(\lambda^*)-Q^{-2}\vt{1} \nu^*.
    \end{alignat}
\end{subequations}
Substituting $\nu^*$ from \eqref{eq:lambda_nu&x_a} in \eqref{kkt:mpec_quad1_c} leads to the scalar linear complementarity relation \eqref{kkt:swsa_reduced} with $\lambda=\lambda^*$. Finally, by applying the change of variable   $z:=\lambda^{\mathrm{max}}- \lambda$,  \eqref{kkt:swsa_reduced} can be written as a standard linear complementarity problem. Existence and uniqueness of the solution to this problem follow from \cite[Theorem~3.1.6]{cottle2009linear}. Thus, $\lambda^*$ is unique. The optimal values of the other primal and dual variables can be determined uniquely from \eqref{eq:lambda_nu&x} and \eqref{kkt:mpec_quad1_middle}.
\end{proof}

\begin{proof}[Proof of Lemma~\ref{lem:dual1}]
Note that $\bar \lambda$ is the optimizer of \eqref{opt:swv_dual} if and only if it satisfies 
\begin{equation} \label{KKT:swv_dual}
    \vt{1}^{\top}Q^{-1}\vt{1} \bar \lambda + (Q^{-1}c_0+a)^{\top}\vt{1}=0.
\end{equation}
Next, we write the necessary and sufficient KKT conditions of \eqref{opt:swv} as
\begin{subequations} \label{KKT:swv}
\begin{align} \label{KKT:swv_a}
    &Q \bar y +c_0 +\vt{1} \bar \lambda=0, \\ \label{KKT:swv_b}
    &\vt{1}^{\top} \bar y- \vt{1}^{\top}a=0.
\end{align}
\end{subequations}
Solving $\bar {y}$ from \eqref{KKT:swv_a} and substituting it in   \eqref{KKT:swv_b}, we obtain $\bar y= -Q^{-1} (\vt{1} \bar \lambda+c_0)$ and \eqref{KKT:swv_dual}. 
\end{proof}

\begin{proof}[Proof of Lemma~\ref{lem:dual2}]
    Let $\bar \mu_{\lambda}$ be the optimal dual variable corresponding to the inequality constraint in \eqref{opt:swv_dual_mod}. Then, $\bar \lambda$ is an optimizer of \eqref{opt:swv_dual_mod} if and only if it satisfies 
    \begin{subequations}
     \label{KKT:swv_dual_mod}
     \begin{align}\label{KKT:swv_dual_mod_a}
         &\vt{1}^{\top}Q^{-1}\vt{1} \bar \lambda + (Q^{-1}c_0+a)^{\top}\vt{1}+\bar \mu_{\lambda}=0, \\ \label{KKT:swv_dual_mod_b}
         &0 \le \bar \mu_{\lambda} \perp\lambda^\max- \bar  \lambda \ge 0.
     \end{align}
    \end{subequations}
    By substituting $\bar \mu_{\lambda}$ from \eqref{KKT:swv_dual_mod_a} in  \eqref{KKT:swv_dual_mod_b}, we get \eqref{kkt:swsa_reduced} with $\lambda=\bar \lambda$. Similarly, $\big((\bar y, \bar s), (\bar \lambda, \bar \mu_s)\big)$ is a primal-dual optimizer of \eqref{opt:swv_primal_mod} if and only if it satisfies
    \begin{subequations}
     \label{KKT:swv_primal_mod}
     \begin{align}\label{KKT:swv_primal_mod_a}
         &Q \bar y+c_0 + \vt{1} \bar \lambda=0, \\ \label{KKT:swv_primal_mod_b}
         & \lambda^\max - \bar \lambda- \bar \mu_s=0,\\ \label{KKT:swv_primal_mod_c}
         &\vt{1}^\top \bar y - \vt{1}^\top  a=\bar s, \\
         \label{KKT:swv_primal_mod_d}
         &0 \le \bar \mu_s \perp \bar s \ge 0.  
     \end{align}
    \end{subequations}
    The equations \eqref{KKT:swv_primal_mod_a} and \eqref{KKT:swv_primal_mod_b} lead to $\bar y=\phi(\bar y)$ and  $\bar \mu_s=\lambda^\max- \bar \lambda$. By substituting $\bar y=\phi(\bar \lambda)$ in \eqref{KKT:swv_primal_mod_c}, we get $\bar s=\vt{1}^{\top}(\phi(\bar \lambda)-a)$. 
     Consequently,  \eqref{KKT:swv_primal_mod} reduces to
     \eqref{kkt:swsa_reduced}. Finally, it follows from Lemma~\ref{lem:swsa_unique}(b) that $\bar \lambda$ and hence, $\big((\bar y, \bar s), (\bar \lambda, \bar \mu_s)\big)$ are unique.
\end{proof}
\begin{proof}[Proof of Proposition~\ref{pro:map}]
Nonsingularity of $M$ follows from the fact that $\det M= \vt{1}^{\top} Q^{-2} \vt{1}>0$. 
Next, it follows from Lemma~\ref{lem:swsa_unique}(a) and  Lemma~\ref{lem:dual2} that $\lambda^*$ and $\bar \lambda$ satisfy the same linear complementarity relation \eqref{kkt:swsa_reduced}. 
Bearing in mind that the solution to \eqref{kkt:swsa_reduced} is unique by Lemma~\ref{lem:swsa_unique}(b), we find that $\bar \lambda=\lambda^*$.
Finally, by \eqref{eq:lambda_nu&x} in the proof of Lemma~\ref{lem:swsa_unique} and $\bar y= \phi(\bar \lambda)$, $\bar s= \vt{1}^{\top}(\phi(\bar \lambda)-a)$ in Lemma~\ref{lem:dual2}, it is straightforward to verify that \eqref{eq:cov} holds.
\end{proof}
\begin{proof}[Proof of Proposition~\ref{pro:unique_equ}]
At any equilibrium point $(\bar x, \bar \rho, \bar \epsilon, \bar \lambda, \bar u, \bar \pi, \bar \nu, \bar \mu)$, we have
\begin{subequations} \label{equ:cl0}
    \begin{alignat}{1} \label{equ:cl0_a}
        0 &= -Q \bar x-c_0-\bar \rho-\bar u, \\  \label{equ:cl0_b}
        0 &= \bar x-a-\bar \epsilon,\\ \label{equ:cl0_c}
        0 &= \bar \rho-\vt{1} \bar \lambda,\\ \label{equ:cl0_d}
        0 &=\vt{1}^{\top} \bar \epsilon, \\ \label{equ:cl0_e}
        0 & =  -Q^{-1} \bar u -Q \bar \pi -\bar x - Q^{-1}(c_0+\lambda^{\mathrm{max}}\vt{1}),\\ \label{equ:cl0_f}
        0 & =  Q \bar  u - \vt{1} \bar \nu,, \\ \label{equ:cl0_g} 
        0 & = \vt{1}^{\top} \bar \pi + \bar \mu,\\ \label{equ:cl0_h} 
        0 &\le \bar \nu \perp \bar \mu \ge 0.
    \end{alignat}
\end{subequations}
Given \eqref{equ:cl0_b}, \eqref{equ:cl0_c} and \eqref{equ:cl0_f}, we can reduce this to
\begin{subequations} \label{equ:cl1}
    \begin{alignat}{1} \label{equ:cl1_a}
        0 &= -Q \bar x-c_0-\vt{1} \bar \lambda-Q^{-1} \vt{1} \bar \nu, \\ \label{equ:cl1_b}
        0 & = \vt{1}^{\top} \bar x- \vt{1}^{\top} a , \\ \label{equ:cl1_c}
        0 & =  -Q^{-2} \vt{1} \bar \nu -Q \bar \pi -\bar x - Q^{-1}(c_0+\lambda^{\mathrm{max}}\vt{1}),\\ \label{equ:cl1_d}
        0 & = \vt{1}^{\top} \bar \pi + \bar \mu,\\ \label{equ:cl1_e} 
        0 &\le \bar \nu \perp \bar \mu \ge 0.
    \end{alignat}
\end{subequations}
We left-multiply both sides of \eqref{equ:cl1_c} by $Q$ and subtract it from \eqref{equ:cl1_a}, which gives 
$\vt{1} (\bar \lambda -\lambda^{\max})=Q^{2} \bar \pi$.
By solving $ \bar \pi$ from the previous equation and substituting the solution in \eqref{equ:cl1_d}, we get $\bar \mu= (\vt{1}^{\top}Q^{-2}\vt{1}) (\lambda^{\max}-\bar \lambda)$. Consequently,  \eqref{equ:cl1} can be further reduced to
\begin{subequations} \label{equ:cl2}
    \begin{alignat}{1} \label{equ:cl2_a}
        0 &= -Q \bar x-c_0-\vt{1} \bar \lambda-Q^{-1} \vt{1} \bar \nu, \\ \label{equ:cl2_b}
        0 & = \vt{1}^{\top} \bar x- \vt{1}^{\top} a , \\ \label{equ:cl2_c}
        0 &\le \bar \nu \perp \lambda^{\max}-\bar \lambda \ge 0,
    \end{alignat}
\end{subequations}
which is equivalent to \eqref{kkt:mpec_quad1}. 
It follows from the proof of Lemma~\ref{lem:swsa_unique} that the point $(\bar x, \bar \lambda, \bar u)$ exists and is unique, namely it is equal to the SCE $(x^*, \lambda^*, u^*)$. Subsequently, given $(\bar x, \bar \lambda, \bar u)=(x^*, \lambda^*, u^*)$, 
the values of $ \bar \rho$, $\bar \epsilon$, $\bar \pi$, $\bar \nu$ and $\bar \mu$ can be uniquely determined from $\eqref{equ:cl0_c}$, $\eqref{equ:cl0_b}$, $\eqref{equ:cl0_e}$, $\eqref{equ:cl0_f}$, $\eqref{equ:cl0_g}$, respectively. 
\end{proof}

\section*{ACKNOWLEDGMENT}

The authors thank Prof. John Simpson-Porco for his thorough review and valuable feedback on the manuscript.


\addtolength{\textheight}{-12cm}   




\bibliographystyle{IEEEtran}
\bibliography{mybib.bib}

\begin{thebibliography}{10}
\providecommand{\url}[1]{#1}
\csname url@samestyle\endcsname
\providecommand{\newblock}{\relax}
\providecommand{\bibinfo}[2]{#2}
\providecommand{\BIBentrySTDinterwordspacing}{\spaceskip=0pt\relax}
\providecommand{\BIBentryALTinterwordstretchfactor}{4}
\providecommand{\BIBentryALTinterwordspacing}{\spaceskip=\fontdimen2\font plus
\BIBentryALTinterwordstretchfactor\fontdimen3\font minus \fontdimen4\font\relax}
\providecommand{\BIBforeignlanguage}[2]{{%
\expandafter\ifx\csname l@#1\endcsname\relax
\typeout{** WARNING: IEEEtran.bst: No hyphenation pattern has been}%
\typeout{** loaded for the language `#1'. Using the pattern for}%
\typeout{** the default language instead.}%
\else
\language=\csname l@#1\endcsname
\fi
#2}}
\providecommand{\BIBdecl}{\relax}
\BIBdecl

\bibitem{nguyen2011walrasian}
D.~T. Nguyen, M.~Negnevitsky, and M.~de~Groot, ``Walrasian market clearing for demand response exchange,'' \emph{IEEE Transactions on Power Systems}, vol.~27, no.~1, pp. 535--544, 2011.

\bibitem{mas1995microeconomic}
A.~Mas-Colell, M.~D. Whinston, J.~R. Green \emph{et~al.}, \emph{Microeconomic theory}.\hskip 1em plus 0.5em minus 0.4em\relax Oxford university press New York, 1995, vol.~1.

\bibitem{li2020transactive}
S.~Li, J.~Lian, A.~J. Conejo, and W.~Zhang, ``Transactive energy systems: The market-based coordination of distributed energy resources,'' \emph{IEEE Control Systems Magazine}, vol.~40, no.~4, pp. 26--52, 2020.

\bibitem{stegink2016unifying}
T.~Stegink, C.~De~Persis, and A.~van~der Schaft, ``A unifying energy-based approach to stability of power grids with market dynamics,'' \emph{IEEE Transactions on Automatic Control}, vol.~62, no.~6, pp. 2612--2622, 2016.

\bibitem{papadaskalopoulos2013decentralized}
D.~Papadaskalopoulos and G.~Strbac, ``Decentralized participation of flexible demand in electricity markets—{P}art {I}: Market mechanism,'' \emph{IEEE Transactions on Power Systems}, vol.~28, no.~4, pp. 3658--3666, 2013.

\bibitem{knudsen2015dynamic}
J.~Knudsen, J.~Hansen, and A.~M. Annaswamy, ``A dynamic market mechanism for the integration of renewables and demand response,'' \emph{IEEE Transactions on Control Systems Technology}, vol.~24, no.~3, pp. 940--955, 2015.

\bibitem{fehr1999theory}
E.~Fehr and K.~M. Schmidt, ``A theory of fairness, competition, and cooperation,'' \emph{The quarterly journal of economics}, vol. 114, no.~3, pp. 817--868, 1999.

\bibitem{gerlagh2022stabilizing}
R.~Gerlagh, M.~Liski, and I.~Vehvil{\"a}inen, ``Stabilizing the {EU} electricity market: Mandatory demand reduction and a lower price cap,'' in \emph{EconPol Forum}, vol.~23, no.~6.\hskip 1em plus 0.5em minus 0.4em\relax Munich: CESifo GmbH, 2022, pp. 8--12.

\bibitem{wei2014competitive}
E.~Wei, A.~Malekian, and A.~Ozdaglar, ``Competitive equilibrium in electricity markets with heterogeneous users and price fluctuation penalty,'' in \emph{53rd IEEE Conference on Decision and Control}.\hskip 1em plus 0.5em minus 0.4em\relax IEEE, 2014, pp. 6452--6458.

\bibitem{chen2022competitive}
Y.~Chen, R.~Islam, E.~L. Ratnam, I.~R. Petersen, and G.~Shi, ``Competitive equilibriums and social shaping for multi-agent systems,'' \emph{Automatica}, vol. 146, p. 110663, 2022.

\bibitem{salehi2023competitive}
Z.~Salehi, Y.~Chen, E.~L. Ratnam, I.~R. Petersen, and G.~Shi, ``Competitive equilibrium for dynamic multi-agent systems: Social shaping and price trajectories,'' \emph{IEEE Transactions on Automatic Control}, 2023.

\bibitem{li2015demand}
N.~Li, L.~Chen, and M.~A. Dahleh, ``Demand response using linear supply function bidding,'' \emph{IEEE Transactions on Smart Grid}, vol.~6, no.~4, pp. 1827--1838, 2015.

\bibitem{arrow1958studies}
K.~Arrow, L.~Hurwicz, and H.~Uzawa, \emph{Studies in linear and non-Linear programming:}.\hskip 1em plus 0.5em minus 0.4em\relax Stanford University Press, 1958.

\bibitem{CHERUKURI201610}
A.~Cherukuri, E.~Mallada, and J.~Cortés, ``Asymptotic convergence of constrained primal-dual dynamics,'' \emph{Systems \& Control Letters}, vol.~87, pp. 10--15, 2016.

\bibitem{boyd2004convex}
S.~P. Boyd and L.~Vandenberghe, \emph{Convex optimization}.\hskip 1em plus 0.5em minus 0.4em\relax Cambridge university press, 2004.

\bibitem{cottle2009linear}
R.~W. Cottle, J.-S. Pang, and R.~E. Stone, \emph{The linear complementarity problem}.\hskip 1em plus 0.5em minus 0.4em\relax SIAM, 2009.

\end{thebibliography}

\end{document}